\documentclass[letterpaper, 10 pt, conference]{ieeeconf}  

\IEEEoverridecommandlockouts                              

\usepackage{graphics} 
\usepackage{amssymb}  
\usepackage{color}
\usepackage{multirow}
\usepackage{bm}
\usepackage{comment}
\usepackage{cite}
\usepackage{amsmath,amssymb}
\usepackage{algorithmic}
\usepackage{graphicx}
\usepackage{epsfig} 
\usepackage{amsmath} 
\usepackage{amssymb}  
\usepackage[bb=boondox]{mathalfa}
\usepackage{algorithm}
\usepackage{subfigure}

\newcommand{\figurename}[1]{Fig.{#1}}

\usepackage{pgfplots}
\pgfplotsset{compat=newest}
\pgfplotsset{plot coordinates/math parser=false}
\newlength\figureheight
\newlength\figurewidth
\usepackage{tikz}
\usetikzlibrary{positioning, shapes, arrows, calc}
\usetikzlibrary{arrows.meta, positioning}

\usepackage[inkscapelatex=false]{svg}
\usepackage{url}

\newtheorem{lemma}{Lemma}
\newtheorem{assumption}{Assumption}
\newtheorem{proposition}{Proposition}
\newtheorem{remark}{Remark}

\newcommand{\tablename}[1]{Table{#1}}

\title{\LARGE \bf
Innovation diffusion dynamics toward long-term behavioral shifts
}
\author{Lisa Piccinin$^{1}$, Valentina Breschi$^{2}$, Chiara Ravazzi$^{3}$, Fabrizio Dabbene$^{3}$ and Mara Tanelli$^{1}$
\thanks{This work is partially supported by PRIN project {\em{TECHIE: A control and network-based approach for fostering the adoption of new technologies in the ecological transition}}. Cod. 2022KPHA24 CUP: D53D23001320006.}
\thanks{$^{1}$Lisa Piccinin and Mara Tanelli are with Politecnico di Milano, 20133, Milan, Italy. 
        {\tt\small name.surname@polimi.it}}%
\thanks{$^{2}$Valentina Breschi is with the Eindhoven University of Technology, 5600 MB, Eindhoven, The Netherlands.
        {\tt\small  v.breschi@tue.nl}}%
\thanks{$^{3}$Chiara Ravazzi and Fabrizio Dabbene are with the Institute of Electronics, Computer and Telecommunication Engineering, National Research Council of Italy (CNR-IEIIT), 10129, Turin, Italy. 
        {\tt\small name.surname@cnr.it}%
}}

\begin{document}

\maketitle
\thispagestyle{empty}
\pagestyle{empty}

\begin{abstract}
Sustainable technologies and services can play a pivotal role in the transition to \textquotedblleft greener\textquotedblright \ habits. Their widespread adoption is thus crucial, and understanding how to foster this phenomenon in a systematic way could have a major impact on our future. With this in mind, in this work we propose an extension of the Friedkin-Johnsen opinion dynamics model toward characterizing the long-term impact of (structural) fostering policies. We then propose alternative nudging strategies that target a trade-off between widespread adoption and investments under budget constraints, showing the impact of our modeling and design choices on inclination shifts over a set of numerical tests. 
\end{abstract}

\section{Introduction}
Promoting the adoption of new green technologies and services has become increasingly important, especially as society faces the pressing challenges of climate change~\cite{c1}. To achieve widespread adoption, policymakers and stakeholders can strategically rely on \emph{on-off interventions} (e.g., first-time user discounts for sharing mobility services), promoting first-hand experiences of the benefits of a service or a technology that can, nevertheless, have only a short-term impact on individual habits~\cite{c3}, and \emph{systemic policies} (e.g., building dedicated parking spaces for sharing vehicles), progressively shaping available infrastructures or regulatory system to facilitate the acceptance of new (green) solutions in the long run~\cite{c4,c20}. However, the effectiveness of these strategies is often hampered (and, thus, the widespread adoption of these green solutions slowed) by changes that the use of new technologies can induce in established habits due to the natural tendency of people to resist changes in their long-standing routines~\cite{c2}, as well as social dictates~\cite{c19}. 

In this context, opinion dynamics and control theory can
be crucial in systematically characterizing and harnessing
the interplay between individual needs, social dictates, and
interventions on personal choices~\cite{c19}. This approach allows benchmarking existing policies and proposing alternative ones, making control strategies key tools in designing human-centered policies to shape a more sustainable
future~\cite{c5}.   

\paragraph*{Modeling individual disposition to adoption} Throughout the years, several mathematical models have been proposed to capture the intricate mechanisms of opinion formation over social networks (see the review in~\cite{c6}). Among them and laying its foundations on~\cite{c7}, the \emph{DeGroot} model (see~\cite{c8}) postulates that personal preferences evolve based on a weighted averaging process dictated by an \textit{influence network}. Accordingly, individuals adjust their inclinations by integrating their beliefs with their neighbors' (weighted) opinions. Several extensions of this framework have then been proposed to explore how diverse social phenomena arise by incorporating more complex dynamics into how individuals learn and assimilate opinions. In particular, the \emph{Friedkin-Johnsen} (FJ) model~\cite{c9} introduces the concept of \textit{stubborn agents}, i.e., individuals inherently less inclined to be influenced, thus describing also the effect of personal biases along with that of social influence in opinion formation.
While characterizing how opinions spread in a social context in open-loop, traditional models have often overlooked the impact of external factors on opinion dynamics. Nonetheless, understanding and characterizing the interplay between opinions and external interventions (apart from social interactions) is key to ensuring the design of efficient and cost-effective interventions, especially in the realm of innovation diffusion.
\paragraph*{Fostering changes in individual attitudes} Traditional opinion dynamics models have often been used to address influence maximization problems, that is, to select an initial set of users to influence the largest number of agents in the network (see, e.g.,~\cite{c23,c24,c25}). 
At the same time, as these models do not explicitly describe the impact of policies on opinion formation, they could not be used for interventions' evaluation and design. Only a few recent studies (see~\cite{c13,c14,c12,c26}) have explored models that embed the effects of external interventions into opinion dynamics to design policies to nudge changes in individual inclinations systematically. Specifically,~\cite{c13}, relying on a DeGroot-like model, assumes that the social network features a \textquotedblleft manipulating agent\textquotedblright \ having access to the inclination of others and analyzes the effect of the manipulating agent using a \emph{proportional-integral} (PI) controller to steer the opinions of others to the desired target. In a similar spirit, but focusing on recommendation systems,~\cite{c12}, builds on~\cite{c10} and~\cite{c11}, extending the FJ model to incorporate the impact of a recommendation system on opinion dynamics. In detail, the recommender is modeled as an additional node in the network and employs model predictive control (MPC) to maximize user engagement while providing indications on the extent to which the recommender can alter users' opinions. Not postulating access to individual opinions,~\cite{c14} assumes instead that the policymaker can nudge individuals by acting on their inherent bias by relying on an estimate of their inclinations, which is used within a centralized MPC scheme to design nudging policies. The work in~\cite{c14} is extended and enhanced in~\cite{c26}, where the evolution of individual inclinations is studied under both controlled interventions and uncontrolled external events, comparing various methods for policy design and their impact on the trade-off between policy costs and widespread diffusion of an opinion.

\paragraph*{Our contributions and goals} 

Despite considering the effect of external interventions, the opinion dynamics models adopted in~\cite{c12, c14, c26} were designed with a focus on on-off incentives. Indeed, these models tend to asymptotically \textquotedblleft forget\textquotedblright \ changes in inclinations induced by nudging policies, with opinions asymptotically converging to their open-loop values. Therefore, they do not allow to model the impact of systemic policies, making it unsuitable for assessing the potential long-term impact of such policies.
To overcome this limitation, we propose an extension of the FJ model with controlled inputs having saturated integral dynamics. This choice allows us to preserve the equivalence with the standard FJ model without external factors driving opinion dynamics apart from social interactions, while enabling us to describe long-term shifts in individual attitudes, mimicking the impact of (well-conceived) structural policies. In a similar spirit to~\cite{c14} and {~\cite{c26}, but assuming access only to individuals' average inclinations over time, we then use the proposed model to devise different strategies for the design of (centralized) fostering policies aimed at achieving an \textquotedblleft optimal\textquotedblright \ \emph{trade-off} between the average adoption of a new technology/service (referred to as \emph{social benefit}) and the investment made (i.e., the policy's \emph{cost}), under budget constraints. This choice allows us to consider a constrained allocation problem while allowing us to (at least empirically) analyze the impact of budget-constrained inputs in closed loop.
\paragraph*{Outline} Section~\ref{sec:model} introduces the proposed opinion dynamics model, whose properties are then analyzed in Section~\ref{sec:prop}. This model is used in Section~\ref{sec:policy_design} to introduce budget-constrained policy design strategies, whose impact on individual opinions is analyzed through a numerical example in Section~\ref{sec:num_example}. The paper ends with some final remarks and directions for future work. 
\paragraph*{Notation} We denote with $\mathbb{N}$, $\mathbb{N}_{0}$ and $\mathbb{R}_{+}$, the set of positive natural numbers, the set of natural numbers including zero, and the set of positive real numbers, respectively. Given any vector $ x \in \mathbb{R}^n $ and matrix $ A \in \mathbb{R}^{m \times n}$, their transposes are denoted as $x^{\top}$ and $A^{\top}$, respectively, while the inverse of $B \in \mathbb{R}^{n \times n}$ is given by $B^{-1}$ and its spectral radius is denoted as $\rho(B)$. The positive (non-negative) definite matrix $A$ is indicated as $A \succ 0$ ($A \succeq 0$). Meanwhile, $\|x\|_{2}$ and $\|A\|_{2}$ denote their 2-norms, $\|x\|_{B}^{2}=x^{\top} Bx$, $x_i \in \mathbb{R}$ denotes the $i$-th component of $x$ and $a_{ij} \in \mathbb{R}$ indicates the element of $A$ in position $(i,j)$. We denote by $\mathbb{1}$ and $\mathbb{0}$ vectors of ones and zeros (of appropriate dimensions) while we indicate identity matrices with $I$. Given a random vector $x \in \mathbb{R}^n$, $\mathbb{E}[x]$ denotes its expected value. The logical operator \emph{or} is indicated as $\lor$. 

\section{Modeling long-term attitude shifts} \label{sec:model}
Consider an influence network with $N \in \mathbb{N}$ agents that we formally characterize with a directed weighted graph $\mathcal{G}=(\mathcal{V}, \mathcal{E},P)$. Here, $\mathcal{V}$ represents the set of $N$ agents, $\mathcal{E}$ indicates the existence of bonds between them, i.e., agent $w\in\mathcal{V}$ influences agent $v\in\mathcal{V}$ if $(v,w) \in \mathcal{E}$, while the strength of mutual influences is dedicated by the non-negative, row-stochastic matrix $P \in \mathbb{R}^{N \times N}$ satisfying       
\begin{equation}\label{eq:properties_P}
    P_{vw}{\color{black}>0},~\forall (v,w) \in \mathcal{E}, \qquad \sum_{w \in \mathcal{V}}  P_{vw}  =  1,~\forall v \in \mathcal{V}.
\end{equation}
The mutual connections of the agents are further shaped through a diagonal matrix $\Lambda \in [0,1]^{N \times N}$, whose diagonal entries $\lambda_{i} \in [0,1]$ denote the susceptibility of the $i$-th agent to the influence of its peers, with $i=1,\ldots,N$. As common in opinion dynamics \cite{c17}, we make the following technical assumption about these susceptibilities.
\begin{assumption}\label{ass:P}
For every node $v \in\mathcal{V}$, there exists a path from $v\in\mathcal{V}$ to a node $w\in\mathcal{V}$ such that $\lambda_w<1$. 
\end{assumption}

Agents are set apart by their \emph{inherent biases} $u^\mathrm{o}   \in   [0,1]^{N}$ to the new technology/services, which we assume are the features that nudging policies can modify. In particular, the closer $u^\mathrm{o}_{v} \in [0,1]$ is to $1$, the more inherently well-disposed the $v$-th agent is to the new technology/service. Along with their biases, agents' are also characterized by their latent inclination to adoption at each time instant $t \in \mathbb{N}_{0}$, here assumed to be collected in a (stochastic) state $ {x}(t) \in [0,1]^{N}$, with $ {x}_{v}(t) \in [0,1]$ closer to $1$ indicating a positive attitude of the $v$-th agent toward the technology/service of interest. The evolution of this variable over time is here characterized by the following difference equation:
\begin{subequations}\label{eq:dyn_model}
\begin{equation}\label{eq:dyn_model_x}
    {x}(t+1)\!=\!\Lambda P  {x}(t)+(I\!-\!\Lambda) \max\{\mathbb{0},\!\min\{u(t)\!+\! u^{nc}(t),\!\mathbb{1}\!\}\!\},
\end{equation}
where $u^{nc}(t)$ characterizes the effect of stochastic (short-term) fluctuations in individual inclination to adoption due to external, uncontrollable factors (e.g., public transport strikes or bad weather in the context of bike sharing) on changes in individual inclinations at time $t\in\mathbb{N}_0$ (as in~\cite{c26}). To reflect this modeling assumption, the uncontrollable inputs thus satisfy the following.
\begin{assumption}\label{ass:noise}
  The uncontrollable inputs in the sequence $\{u^{nc}(t)\}_{t\in\mathbb{N}_0}$ are i.i.d. random vectors uniformly distributed within the interval $ [- \delta, +\delta] $, with $0 < \delta < u^{\mathrm{o}}_v$ for all $v \in \mathcal{V}$.
\end{assumption}
Due to the condition on $\delta$, we ultimately assume that $u^{nc}(t)$ causes only (very) slight short-term changes in individual opinions at all $t \in \mathbb{N}_0$, since $|u^{nc}_v(t)|\leq \delta$ for all $t \in \mathbb{N}_0$ and $v \in \mathcal{V}$.

Meanwhile, $u(t)$ encodes the impact of the initial individual bias and the cumulative ones of policy actions enacted by a policymaker or a stakeholder until $t\in\mathbb{N}_0$. Specifically, we describe the evolution of $u(t)$ with the following (simplistic) cumulative dynamics: 
\begin{equation}\label{eq:input_dynamics}
  u(t+1)=
        u(t)+u^{c}(t)
    ,~~\forall t \in \mathbb{N},~~\mathrm{with}~u(0)=u^{\mathrm{o}},
\end{equation}
where $u^{c}(t)$ is a non-negative controlled input modeling actions that the policymaker or stakeholders can undertake (and adjust) over time to nudge a shift in individual preferences toward a new technology/service.
\end{subequations}
\begin{remark}[The validity of \eqref{eq:dyn_model}]
    Including the \textit{max/min} in \eqref{eq:dyn_model_x} ensures that $x(t)$ is well defined $\forall t \in \mathbb{N}_{0}$, i.e., that $x(t) \in [0,1]^{N}$ at all time instants. \hfill $\square$
\end{remark}

It is worth pointing out that, when no controlled policy is deployed to nudge the acceptance of a new technology/service, the latent state's expected value asymptotically coincides with that of the standard FJ model, as formalized in the following lemma.
\begin{lemma}[Control-free mean asymptotic opinions]\label{lemma:steady_free}
    Let Assumption~\ref{ass:P} hold, $x(0) \in [0,1]^{N}$ and $u^{c}(t)=0$ for all $t \in \mathbb{N}_0$. Then, the latent state's expected value satisfies
    \begin{equation}\label{eq:steady_state_free}
    \mu_{\infty}\!=\!\lim_{t \rightarrow \infty} \mu(t)\!=\!\lim_{t \rightarrow \infty} \mathbb{E}[x(t)]=(I-\Lambda P)^{-1}(I-\Lambda)u^{\mathrm{o}},
    \end{equation}
    with $\mu_{\infty} \in [0,1]^{N}$.
\end{lemma}
\begin{proof}
    Since $u^{c}(t)=0$ for all $t \in \mathbb{N}_0$, according to \eqref{eq:input_dynamics} then $u(t+1)=u(t)=u^{\mathrm{o}} \in [0,1]^{N}$ for all $t \in \mathbb{N}_{0}$. In turn, it thus straightforwardly follows from \eqref{eq:dyn_model} that
    \begin{equation*}
        x(t+1)=\Lambda P x(t)+(I-\Lambda)(u^{\mathrm{o}}+u^{nc}(t)),~~\forall t \in \mathbb{N}_0,
    \end{equation*}
    and, accordingly, that 
    \begin{equation}\label{eq:expected_dynamics}
    \mu(t+1):=\mathbb{E}[ {x}(t+1)]=\Lambda P \mu(t)+(I-\Lambda)u^{\mathrm{o}},~~\forall t \in \mathbb{N}_0.
    \end{equation}
    The steady-state result in \eqref{eq:steady_state_free} straightforwardly results from the same reasoning in \cite{c17}, concluding the proof.
\end{proof}
Therefore, under Assumption~\ref{ass:P}, the expected latent opinions converge to a profile that is a convex combination of the initial inclinations of the agents driven by the strength of their mutual bonds and their susceptibility.
\begin{remark}[Lemma~\ref{lemma:steady_free} and our modeling choices]
    The result in Lemma~\ref{lemma:steady_free} indicates that persistent changes in opinions can be achieved by acting on individual biases, supporting our choice of designing fostering policies directly targeting a change in $u^{\mathrm{o}}$ and indirectly leveraging social imitation, while not changing the set of initial adopters nor changing the features of social bonds. 
\end{remark}
\subsection{Interpreting our model as a multilayer one}
\begin{figure}[!tb]
    \centering
\begin{tikzpicture} 
\scalebox{.8}{\begin{scope}[yshift=-1cm]
\path (1.2,5.75) edge[black, -{Latex}] (1.2,4.75);
\node[fill=white, draw=white, rectangle] at (2.8, 5.25) {White noise};

\fill[pink!30] (-1,4) -- (2.5,4) -- (1.75,4.75) -- (-1.75,4.75) -- cycle;

\node[circle, fill=pink!90!black, inner sep=1.5pt] (A1) at (0,4.375) {};
\node[circle, fill=pink!90!black, inner sep=1.5pt] (A2) at (-0.8,4.5) {};
\node[circle, fill=pink!90!black, inner sep=1.5pt] (A3) at (1,4.25) {};
\node[circle, fill=pink!90!black, inner sep=1.5pt] (A4) at (1.5,4.5) {};

\node[fill=white, draw=white, rectangle ] at (4.05, 4.25) {Layer 4: \textit{Jammer}};

\fill[cyan!15] (-1,2.75) -- (2.5,2.75) -- (1.75,3.5) -- (-1.75,3.5) -- cycle;

\node[circle, fill=cyan, inner sep=1.5pt] (B1) at (0,3.125) {};
\node[circle, fill=cyan, inner sep=1.5pt] (B2) at (-0.8,3.25) {};
\node[circle, fill=cyan, inner sep=1.5pt] (B3) at (1,3.0) {};
\node[circle, fill=cyan, inner sep=1.5pt] (B4) at (1.5,3.25) {};

\node[fill=white, draw=white, rectangle ] at (4, 3.125) {Layer 3: \textit{Agents}};

\fill[teal!20] (-1,1.5) -- (2.5,1.5) -- (1.75,2.25) -- (-1.75,2.25) -- cycle;

\node[circle, fill=teal, inner sep=1.5pt] (C1) at (0,1.875) {};
\node[circle, fill=teal, inner sep=1.5pt] (C2) at (-0.8,2) {};
\node[circle, fill=teal, inner sep=1.5pt] (C3) at (1,1.75) {};
\node[circle, fill=teal, inner sep=1.5pt] (C4) at (1.5,2) {};

\node[fill=white, draw=white, rectangle ] at (4.2, 1.875) {Layer 2: \textit{Reservoir}};

\path (C1) edge[teal!80, -{Latex}, looseness=10, out=130, in=350] (C1);
\path (C2) edge[teal!80, -{Latex}, looseness=10, out=130, in=350] (C2);
\path (C3) edge[teal!80, -{Latex}, looseness=10, out=130, in=350] (C3);
\path (C4) edge[teal!80, -{Latex}, looseness=10, out=130, in=350] (C4);

\fill[gray!20] (-1,0) -- (2.5,0) -- (1.75,0.75) -- (-1.75,0.75) -- cycle;

\node[circle, fill=gray, inner sep=1.5pt] (D1) at (0,0.375) {};
\node[circle, fill=gray, inner sep=1.5pt] (D2) at (-0.8,0.5) {};
\node[circle, fill=gray, inner sep=1.5pt] (D3) at (1,0.25) {};
\node[circle, fill=gray, inner sep=1.5pt] (D4) at (1.5,0.5) {};

\node[fill=white, draw=white, rectangle] at (4.4, 0.375) {Layer 1: \textit{Policymaker}};

\path (B1) edge[cyan!90, -{Latex}, bend right=20] (B3);
\path (B2) edge[cyan!90, -{Latex}, bend right=20] (B1);
\path (B3) edge[cyan!90, -{Latex}, bend right=30] (B2);
\path (B4) edge[cyan!90, -{Latex}, bend right=15] (B3);
\path (B3) edge[cyan!90, -{Latex}, bend right=15] (B4);

\path (B2) edge[cyan!90, -{Latex}, looseness=10, out=-120, in=-200] (B2);
\path (B1) edge[cyan!90, -{Latex}, looseness=10, out=350, in=50] (B1);
\path (B3) edge[cyan!90, -{Latex}, looseness=10, out=10, in=-50] (B3);
\path (B4) edge[cyan!90, -{Latex}, looseness=10, out=-50, in=10] (B4);

\draw[pink!90!black, -{Latex}] (A1) -- (B1);
\draw[pink!90!black,  -{Latex}] (A2) -- (B2);
\draw[pink!90!black, -{Latex}] (A3) -- (B3);
\draw[pink!90!black, -{Latex}] (A4) -- (B4);

\draw[teal, -{Latex}] (C1) -- (B1);
\draw[teal, -{Latex}] (C2) -- (B2);
\draw[teal, -{Latex}] (C3) -- (B3);
\draw[teal, -{Latex}] (C4) -- (B4);

\draw[gray,-{Latex}] (D1) -- (C1);
\draw[gray,-{Latex}] (D2) -- (C2);
\draw[gray,-{Latex}] (D3) -- (C3);
\draw[gray,-{Latex}] (D4) -- (C4);

\path (2,3.25) edge[black!80, -{Latex}, dashed, bend left=30] (2,0.5);

\path (1.2,-1) edge[black, -{Latex}] (1.2,0);
\node[fill=white, draw=white, rectangle] at (4, -0.75) {Market Demands/Societal needs};

\end{scope}}
\end{tikzpicture}\vspace{-.2cm}
    \caption{A multilayer representation of the opinion dynamics model in~\eqref{eq:dyn_model}.}
    \label{fig:multi}
\end{figure}
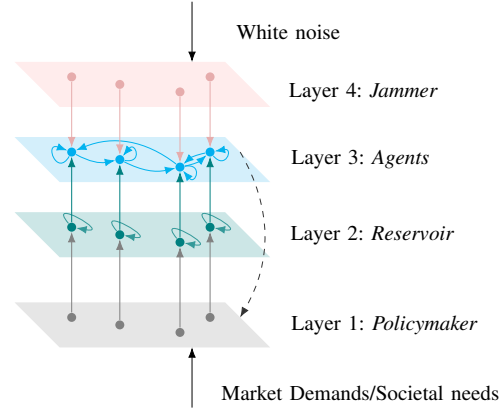
As schematized in \figurename{~\ref{fig:multi}}, our opinion dynamics model in \eqref{eq:dyn_model} ultimately incorporates four interacting layers, each associated with a main actor of our framework. 

As policymakers generally nudge agents' opinions and behaviors in a desired direction based on \textquotedblleft external dictates\textquotedblright, like societal needs or market demands, the \emph{policymaker layer} shapes its actions at each time step $t \in \mathbb{N}_0$ in response to these dictates through the controllable input $u^c(t)$. The latter is transmitted to the \emph{reservoir layer}, which accumulates these actions according to \eqref{eq:input_dynamics} and communicates them to the \emph{agents layer}, ultimately driving agents' opinions through such interaction. In turn, agents' inclinations dynamically evolve not only due to their mutual interactions and the direct exchanges with the reservoir layer\footnote{Implying indirect interactions with the policymaker layer too.}, but also to uncontrolled external factors (see \eqref{eq:dyn_model_x}). The impact of these additional exogenous actors could be represented through a final layer, the \emph{jammer's layer},  modeled as a source of white (and, hence, unpredictable) noise entering the system and introducing random fluctuations. The jammer's influence is unpredictable, representing the role of uncontrollable, external influences in the system's dynamics.

When the policymaker relies on a feedback mechanism to continuously adjust the fostering policies according to the agents' needs, this choice additionally introduces an interaction between the agents and the policymaker layers (represented as a dashed line in \figurename{~\ref{fig:multi}}).  

\section{A closer look at the model's properties} \label{sec:prop}
We now analyze the properties of the opinion dynamics model in \eqref{eq:dyn_model} for two classes of interventions, namely static and feedback policies. 

In both cases and in line with Assumption~\ref{ass:noise}, we suppose that the policies $\{u^{c}(t)\}_{t \in \mathbb{N}_0}$ satisfy the following (worst-case) constraint by design\footnote{In formulating our policy design problems (see Section~\ref{sec:policy_design}), we explicitly impose this constraint.}:
\begin{equation} \label{eq:bound_u}
\mathbb{0} \leq \delta - u^\mathrm{o} \leq \sum_{\tau=0}^{t-1} u^c(\tau) \leq \mathbb{1} - \delta - u^\mathrm{o}, 
\end{equation}
to ensure that $u(t)+u^{nc}(t) \in [0,1]^{N}$ for all possible realizations of the stochastic input $u^{nc}(t)$ at all $t \in \mathbb{N}_0$. According to \eqref{eq:bound_u}, the system dynamics \eqref{eq:dyn_model} can thus equivalently be rewritten as 
\begin{equation} \label{eq:eq_dyn_cons}
{x}(t+1)=\Lambda P  {x}(t)+(I-\Lambda)(u(t)+u^{nc}(t)).
\end{equation}
Apart from assuming that \eqref{eq:bound_u} holds, we further assume that the designed policies are deployed under budget constraints, as formalized in the following assumption.
\begin{assumption}[Limited resources]\label{ass:budget}
    Nudging policies are enacted under a fixed and finite \emph{budget} $\beta \in \mathbb{R}_{+}$, with $\beta \ll \infty$, depleted over time. Therefore, the resources available at time $t \in \mathbb{N}_0$ to nudge individuals are dictated by
\begin{equation}\label{eq:budget}
    U(t)=\max\left\{0,\beta-\sum_{k=0}^{t-1}\sum_{v \in \mathcal{V}}u_{v}^{c}(t-k)\right\}.
\end{equation}
\end{assumption}
\subsection{Expected inclinations with constant policies}\label{sec:constant_policy_ours}
Let us first consider the following static policy: 
\begin{equation}\label{eq:constant_input}
    u^{c}(t)\!=\!\begin{cases}
        \nu, ~~
        \text{if}~ u(t)\!+\! u^c(t) \!\in\! [0,\!1]^{N} \lor U(t) \!\geq\! \sum_{v\in\mathcal{V}}\nu_v,\\
        \nu^{\mathrm{r}}, ~ 
        \text{if}~ u(t)\!+\! u^c(t) \!\in\! [0,\! 1]^{N} \lor U(t) \!\in\!(0,\!\sum_{v\in\mathcal{V}}\!\nu_v),\\
        0, ~ \mbox{ otherwise.}
    \end{cases}
\end{equation}
where $\nu_v  \neq  0$ is the baseline magnitude of the intervention the policymaker has tailored to the $v$-th agent, with $v \in \mathcal{V}$, while $\nu^{\mathrm{r}}=\nu\frac{U(t)}{\sum_{v\in\mathcal{V}}\nu_v}$ so that the baseline input $\nu$ is proportionally scaled to deplete remaining resources in one step within the feasibility limits\footnote{Alternative definitions of $\nu^\mathrm{r}$ can also be considered, yet not changing our formal results.}. Note that, if it is not possible to exploit the whole budget without exceeding such limits, some resources can still remain unused. Meanwhile, when $\nu$ is proportional to the budget $\beta$, then $\nu^{\mathrm{r}}=\mathbb{0}$.

Let us then assume that the following holds.
\begin{assumption}\label{ass:finite_time_nu}
The static policy $u^c(t)=\nu$ is enacted until a finite instant $T \!\in\! \mathbb{N}$, i.e., $u^c(T\!+\!1)\!=\!\nu^{\mathrm{r}}$ and $u^c(t)\!=\!\mathbb{0}$ for all $t >T\!+\!1$.
\end{assumption}
By relying on this assumption, we can formalize an asymptotic result on the expected latent inclinations.
\begin{proposition}[Asymptotic opinions under static policies]\label{prop:asymptot_ours}
Let Assumptions~\ref{ass:P}-\ref{ass:finite_time_nu} and \eqref{eq:bound_u} hold. Then the asymptotic expected inclinations under the policy in \eqref{eq:constant_input} satisfy
\begin{subequations}
   \begin{equation}\label{eq:equilibria}
       \mu_{\infty}=(I-P\Lambda)^{-1}(I-\Lambda)\bar{u},
   \end{equation} 
   with $\mu_{\infty} \in [0,1]^{N}$,
   \begin{equation}\label{eq:u_bar_def}
       \bar{u}=u^{\mathrm{o}}+T\nu +\nu^{r},
   \end{equation}
   and $\bar{u}\ll \infty$.
\end{subequations}
\end{proposition}
\begin{proof}
 Thanks to~\eqref{eq:bound_u}, the latent inclination evolves according to~\eqref{eq:eq_dyn_cons} and, thus, it is straightforward to prove that
 \begin{equation}\label{eq:average_state_New}
     \mu(t+1):=\mathbb{E}[x(t+1)]=\Lambda P \underbrace{\mathbb{E}[x(t)]}_{:=\mu(t)}+(I-\Lambda)\mathbb{E}[u(t)],
 \end{equation}
 due to Assumption~\ref{ass:noise}. Meanwhile, because of Assumption~\ref{ass:finite_time_nu}, then 
 \begin{equation}
 \mathbb{E}[u(t)]=\bar{u}=u^{\mathrm{o}}+T\nu +\nu^{r},~~ \forall t > T+1,
\end{equation} 
where $\bar{u}\ll \infty$ is the maximum achievable value for $\mathbb{E}[u(t)]$ under the enacted constant (yet saturated) policy in \eqref{eq:constant_input}, with its finiteness being a consequence of \eqref{eq:bound_u} and Assumption~\ref{ass:budget}. Accordingly, it further holds that
\begin{equation}\label{eq:condition_input_constant2}
\lim_{t\rightarrow \infty} \mathbb{E}[u(t)]=\bar{u},     
\end{equation}
from which \eqref{eq:equilibria} follows thanks to Assumption~\ref{ass:P}, thus concluding the proof.
\end{proof}
Therefore, expected inclinations toward a new technology/service asymptotically converge to a finite value dictated by the characteristics of the static intervention in \eqref{eq:constant_input}, the features of interpersonal bonds, and the agents' initial biases.
\subsection{Comparison with~\cite{c26} under static policies}
Toward showing the suitability of the proposed model to characterize long-term shifts in individual attitudes, we now compare the asymptotic expected inclinations reported in~\eqref{eq:equilibria} with those attained by using the model proposed in~\cite{c26} under the assumption that the policymaker enacts \eqref{eq:constant_input}. To this end, let $\mu^{\mathrm{st}}(t)$ be the mean inclinations at time $t \in \mathbb{N}_0$ dictated by the model proposed in~\cite{c26}, which evolves according to
\begin{equation}\label{eq:average_state_16}
    \mu^{\mathrm{st}}(t+1)=\Lambda P \mu^{\mathrm{st}}(t)+(I-\Lambda)(u^{\mathrm{o}}+u^{c}(t)).
\end{equation}
Then, the following asymptotic result holds.
\begin{proposition}\label{prop:asymp_16}
    Let Assumptions~\ref{ass:P}-\ref{ass:finite_time_nu} and \eqref{eq:bound_u} hold. Then, the mean inclinations in \eqref{eq:average_state_16} under the policy in \eqref{eq:constant_input} satisfy
   \begin{equation}\label{eq:asymptotic_constant_input}
       \mu_{\infty}^{\mathrm{st}}=\lim_{t \rightarrow \infty} \mu^{\mathrm{st}}(t)=(I-P\Lambda)^{-1}(I-\Lambda)u^{\mathrm{o}},
   \end{equation} 
   with $\mu_{\infty}^{\mathrm{st}} \in [0,1]^{N}$.
\end{proposition}
\begin{proof}
    According to \eqref{eq:constant_input}, 
    the following holds:
    \begin{equation*}
        u^{\mathrm{o}}+u^{c}(t)=\begin{cases}
            u^{\mathrm{o}}+\nu,&\text{if } t \leq T,\\
            u^{\mathrm{o}}+\nu^{\mathrm{r}},& \text{if } t = T +1,\\
            u^{\mathrm{o}},& \text{if } t > T+1.
        \end{cases}
    \end{equation*}
    Therefore, 
    it easily follows that 
    \begin{equation*}
        \lim_{t\rightarrow \infty} u^{\mathrm{o}}+u^c(t)=u^{\mathrm{o}},
    \end{equation*}
    and, consequently, that \eqref{eq:asymptotic_constant_input} holds, thus ending the proof.
\end{proof}
As \eqref{eq:asymptotic_constant_input} coincides with the asymptotic latent inclinations in the absence of external, controlled inputs, this result highlights that the model proposed in \cite{c26} implicitly relies on the assumption that any policy enacted with a limited budget cannot lead to an irreversible shift in one's expected inclination. Note that, while this might reflect reality when on-off policies are undertaken, such an (implicit) assumption is instead likely falsified when systemic actions are performed.

Based on Proposition~\ref{prop:asymp_16}, we can then compare the asymptotic inclinations resulting from our modeling choices and those made in \cite{c26}, as subsequently formalized. 
\begin{proposition}\label{prop:comparison}
    Let Assumptions~\ref{ass:P}-\ref{ass:finite_time_nu} and \eqref{eq:bound_u} hold. Let $\mu(t)$ and $\mu^{\mathrm{st}}(t)$ evolve as in \eqref{eq:average_state_New} and \eqref{eq:average_state_16}, respectively. Then, by enacting the policy in \eqref{eq:constant_input}, it asymptotically holds that
    \begin{equation}
        \mu^{\mathrm{st}  }(\infty)<\mu(\infty).
    \end{equation}
\end{proposition}
\begin{proof}
    The proof straightforwardly follows from the definition of $\bar{u}$ in \eqref{eq:u_bar_def} and it is thus omitted.
\end{proof}
\subsection{Feedback policies with constraints and budget limitations}
Let us now consider a static feedback policy of the error between full acceptance and the average individual inclinations, i.e.,
\begin{subequations}\label{eq:state_feedback}
\begin{equation}
     {u^{c}}(t)=K(\mathbb{1}-\mu(t)), 
\end{equation}
with $K  \in  \mathbb{R}^{N \times N}$ designed such that
\begin{equation}
    \rho(\Lambda P  -  (I  -  
  \Lambda)K)<1.
\end{equation}
\end{subequations}
It is worth remarking that \eqref{eq:state_feedback} implies that ${u^{c}}(t)=\mathbb{0}$ whenever $\mu(t)=\mathbb{1}$ and, instead, ${u^{c}}(t)=\mathbb{1}$ if $\mu(t)=\mathbb{0}$. Let us further assume the following.
\begin{assumption}[Bound on the feedback policy]\label{ass:feedback}
    The feedback policy $u^{c}(t)$ in \eqref{eq:state_feedback} is bounded by design in an interval $[u_{\mathrm{min}}^{c},u_{\mathrm{max}}^{c}]$, such that that $\mu(t) \in [0,1]^{N}$ for all $t \in \mathbb{N}_0$.  
\end{assumption}
Hence, the mean cumulative input $\mathbb{E}[u(t)]$ satisfies 
\begin{equation}\label{eq:limits_behavior}
u^{\mathrm{o}}+\mathcal{T}_{\mathrm{min}}u_{\mathrm{min}}^{c}\leq \mathbb{E}[ {u}(t)] \leq u^{\mathrm{o}}+\mathcal{T}_{\mathrm{max}}u_{\mathrm{max}}^{c},~~~\forall t \in \mathbb{N}_0,
\end{equation}
where $\mathcal{T}_{\mathrm{min}}$ and $\mathcal{T}_{\mathrm{max}}$ are diagonal matrices containing the time instants after which the controlled input is set to zero due to saturation of the states or consumption of the budget (similarly to \eqref{eq:u_bar_def}). 
Accordingly, the following asymptotic result holds.
\begin{proposition}[Asymptotic opinions and feedback] Let Assumptions~\ref{ass:P}-\ref{ass:budget} be satisfied and let the enacted policy be defined as in \eqref{eq:state_feedback} while satisfying Assumption~\ref{ass:feedback}.    Then, the expected inclination achieved by closing the loop is asymptotically limited, i.e.,
    \begin{equation}\label{eq:limited_opinions}
        \mu_{\infty}:=\lim_{t \rightarrow \infty} \mathbb{E}[ {x}(t)] \ll \infty.
    \end{equation}
\end{proposition}
\begin{proof}
Along the same line of the proof of Proposition~\ref{prop:asymptot_ours}, the result in~\eqref{eq:limited_opinions} straightforwardly follows from \eqref{eq:limits_behavior}. Therefore, the proof is omitted.
\end{proof}

\section{Toward optimal nudging policy design}\label{sec:policy_design}
By relying on the model introduced in Section~\ref{sec:model}, we now propose two strategies that policymakers/stakeholders can adopt to design interventions that trade-off encouraging the widespread adoption of a new technology/service (i.e., maximizing social benefit) and avoiding waste of resources. 
\subsection{Optimized Constant Control Policy (CCP)}
As a first alternative to design a policy that targets the aforementioned goal in one shot, policymakers/stakeholders can take advantage of the asymptotic properties of the proposed model (discussed in~Section~\ref{sec:constant_policy_ours}). Specifically, considering a prefixed time horizon $T \in \mathbb{N}$ for the policy's deployment, a constant policy can be designed by solving the following problem
\begin{subequations}\label{eq:CCP_opt}
  \begin{align}
  \underset{\mu_{\infty},~u_{\infty}^{c}}{\mathrm{minimize}}&~~J^{\mathrm{CCP}}(\mu_{\infty},u_{\infty}^{c})\\
      \mbox{s.t.} &~~\mu_{\infty}=(I-\Lambda P)^{-1}(I-\Lambda)(u^{\mathrm{o}}+Tu_{\infty}^c),\\
      &~~T\sum_{v \in \mathcal{V}}u_{\infty,v}^c\leq \beta,\\
      & ~~u_{\infty,v}^c \geq 0, ~~\forall v \in \mathcal{V},\\
      &~~u^{\mathrm{o}}_{v}+Tu_{\infty,v}^c\leq 1-\delta,~~\forall v \in \mathcal{V},
  \end{align}  
  where the last constraint guarantees that \eqref{eq:bound_u} is satisfied, and the loss is defined as 
  \begin{equation}
      J^{\mathrm{CCP}}(\mu_{\infty},u_{\infty}^{c})\!=\! \|\mathbb{1}\!-\!\mu_{\infty}\|_{2}^{2\!}+\!\|Tu_{\infty}^c\|_{R}^{2\!}+\!\left\|\beta\!-\!T\!\sum_{v\in \mathcal{V}}\!u_{\infty,v}^c\right\|_{S}^{2\!}\!\!\!,
  \end{equation}
  with $R\succ 0$ and $S \succeq 0$ being penalties chosen by the policymaker, and the last term in the cost aims at minimizing the amount of unused resources. Note that, since $\delta < u^{\mathrm{o}}_v$ for all $v \in \mathcal{V}$ by Assumption~\ref{ass:noise}, the lower-bound in~\eqref{eq:bound_u} is guaranteed by construction.
\end{subequations}
As a result, then policymakers can enact
\begin{equation}
    u^{c}(t)=u_{\infty},~~\forall t \in \{0,1,\ldots,T-1\}.
\end{equation}
\subsection{Model Predictive Control (MPC) Fostering Policy}\label{sec:MPC}
Instead of looking at asymptotic behaviors, a policymaker/stakeholder can instead decide to optimize its strategies in a receding horizon fashion\footnote{
This choice allows us to mitigate the (strong) requirement of an infinite horizon control policy, which necessitates policymakers to unrealistically foreseen individual average attitudes over an infinite time span.} by monitoring average opinions and accordingly adjusting their strategies over time. In this case, the control problem can be formulated as:
\begin{subequations}\label{eq:MPC}
    \begin{align} 
      &\underset{\mathcal{M}_{|t},\mathcal{U}_{|t}^{c}}{\mathrm{minimize}}~~J^{\mathrm{MPC}}(\mathcal{M}_{|t},\mathcal{U}_{|t}^{c})\\  
        &\quad ~~~\mbox{s.t.}~~\mu_{|t}(k\!+\!1)\!=\!\Lambda P \mu_{|t}(k)\!+\!(I\!-\!\Lambda)u_{|t}(k),\\
        & \quad \quad \quad~~~ u_{|t}(k\!+\!1)\!=\!u_{|t}(k)\!+\!u_{|t}^{c}(k),~k \!\in\! [0,L\!-\!1],\\
         & \quad \quad \quad~~~u_{|t}^{\Sigma}(k)=\sum_{\tau=0}^{k}\sum_{v \in \mathcal{V}}u_{|t}^{c}(\tau),~~k \!\in\! [0,L\!-\!1],\\
        &\quad \quad \quad~~~u_{|t,v}^{c}(k) \geq 0,~\forall v \in \mathcal{V},~k\!\in\![0,L\!-\!1],\\
        &\quad \quad \quad~~~u_{|t}^{\Sigma}(k)\leq U(t)\!-\!u_{|t}^{\Sigma}(k\!-\!1),~k \in [1,L\!-\!1],\label{eq:budget_constr_MPC}\\
        &\quad \quad \quad~~~u_{|t,v}(k)\leq 1-\delta,~\forall v \in \mathcal{V},~k\!\in\![0,L\!-\!1],\label{eq:feasibility_constr_MPC}\\
       &\quad \quad \quad~~~ u_{|t}(0)\!=\!u(t),~~~ \mu_{|t}(0)=\mu(t),
      \end{align}
    with $\mathcal{M}_{|t}\!=\!\{\mu_{|t}(k)\}_{k=0}^{L}$ and $\mathcal{U}_{|t}^{c}\!=\!\{u_{|t}^{c}(k)\}_{k=0}^{L-1}$, $U(t)$ being defined as in \eqref{eq:budget}, and 
    \begin{align}\label{eq:MPC_cost}
        \nonumber J^{\mathrm{MPC}}(\mathcal{M}_{|t},\mathcal{U}_{|t}^{c})=&\sum_{k=0}^{L-1} \|\mathbb{1}\!-\!\mu_{|t}(k)\|_{2}^{2}+\|u^{c}_{|t}(k)\|_{R}^{2}\\
        &\qquad \qquad  +\|\mathbb{1}-\mu_{|t}(L\!-\!1)\|_{Q}^{2},
    \end{align}
    where $L \geq 1$ is the prediction horizon decided by the policymaker/stakeholder, while $R \succ 0$ controls the trade-off between adoption boosting/cost containment and the terminal penalty can be weighed according to 
    \begin{equation}
    (\Lambda P)^{\top} Q\Lambda P-Q=-I,
    \end{equation}
    since the dynamics of expected opinions is asymptotically stable by Assumption~\ref{ass:P}. Note that, while \eqref{eq:budget_constr_MPC} allows us to explicitly account for budget consumption in policy design (see Assumptions~\ref{ass:budget}), \eqref{eq:feasibility_constr_MPC} guarantees that \eqref{eq:bound_u} is satisfied as the lower-bound is already verified by construction (see the initial condition in \eqref{eq:input_dynamics}).
\end{subequations}
It is worth remarking that the cost in \eqref{eq:MPC_cost} comprises two terms that penalize the average distance of the agents' opinions to the acceptance of the targeted technology/service and the second term, which weights the distance of the policy action to be designed from its (equilibrium) value at an average full adoption\footnote{It is straightforward to prove that the controlled input at full adoption is $\bar{u}^{c}=\mathbb{0}$ and, thus, such proof is omitted.}.  

\begin{remark}[Practical issues with policy implementation]
Designing a policy as in \eqref{eq:MPC} requires continuative monitoring of the average agents' inclination, which is likely unfeasible in practice. We postpone tackling this practical issue along the same lines of \cite{c26} in future works.     
\end{remark}

\section{Numerical Example} \label{sec:num_example}
\begin{figure}[!tb]
    \centering
    \includegraphics[scale=0.5,trim=0cm 2cm 0cm 1.75cm,clip]{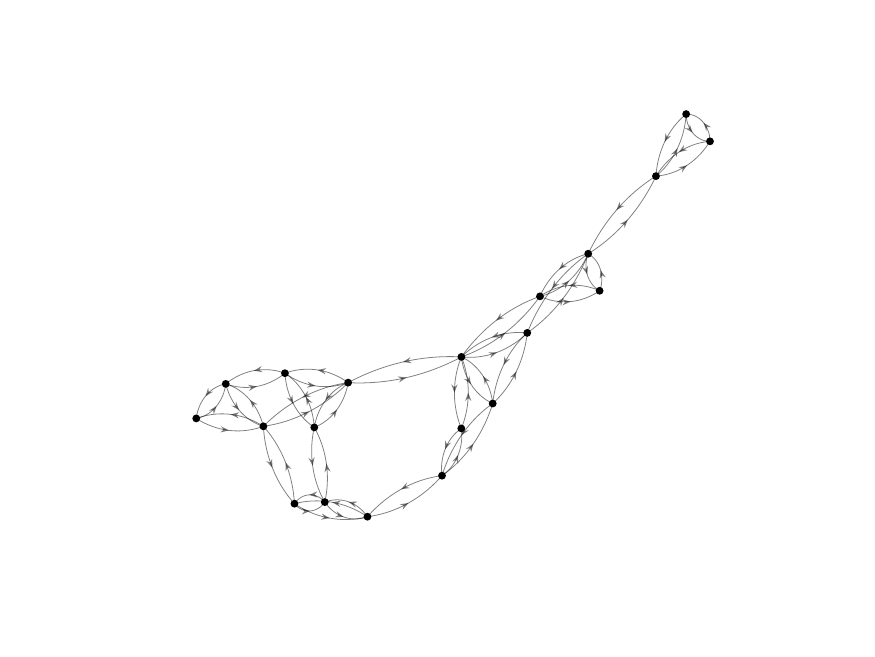}\vspace{-.3cm}
    \caption{The social network considered in our example, featuring $20$ agents and $7$ clusters of agents.}\vspace{-.5cm}
    \label{fig:grafo1}
\end{figure}
We now analyze the impact of the strategies introduced in Section~\ref{sec:policy_design} considering the (randomly generated) modular social network depicted in \figurename{~\ref{fig:grafo1}} and assuming that individual opinions evolve according to the long-term shifts model proposed in Section~\ref{sec:model}.

The considered social network comprises $N=20$ agents and $7$ clusters, generated by setting a link probability of $0.2$ and the probability of connection between agents of different clusters at $0.7$. We impose the elements in $\Lambda$ (see \eqref{eq:dyn_model}) to be equal, i.e., $\Lambda=\lambda I$, yet considering two scenarios where external interventions (i.e., $\lambda=0.25$) and social influences (namely, $\lambda=0.75$) are the main drivers of adoption, respectively. Meanwhile, we consider three setups for the initial bias $u^{\mathrm{o}}$ by splitting the agents into two groups, namely 
\begin{enumerate}
    \item \emph{mixed} biases: $u^{\mathrm{o}}_{v}=0.2$ for $10$ agents and $0.8$ for the remaining ones, thus having a population grouped into distinct factions with respect to the new technology/service;  
    \item \emph{negative} biases: $u^{\mathrm{o}}_{v}=0.2$ for half agents and $0.3$ for the remaining ones, so that the considered population has a negative polarization toward the technology/service;
    \item \emph{positive} biases: $u^{\mathrm{o}}_{v}=0.6$ for $10$ agents and $0.8$ for the remaining ones, considering a population that is instead positively polarized toward such a technology/service.
\end{enumerate}
The performance achieved through the designed policies is evaluated by looking at their \emph{social benefit}, here defined as 
\begin{equation}\label{eq:gammaT}
    \Gamma_{\mathrm{sim}}=\|\mathbb{1}-x(T_{\mathrm{sim}})\|_{2}^{2},
\end{equation}
where $T_{\mathrm{sim}}=30$ is the considered simulation horizon, as well as their cumulative cost and usage of the available budget, i.e.,
\begin{equation}\label{eq:cumulative_cost}
    u_{\mathrm{sim}}^{\Sigma}=\sum_{t=0}^{T_{\mathrm{sim}}-1}\sum_{v \in \mathcal{V}}u_{v}^{c}(t),~~
    B_{\%}=100\frac{u_{\mathrm{sim}}^{\Sigma}}{\beta}~[\%],
\end{equation}
respectively. In all our tests, uncontrollable factors are described as realizations of a uniformly distributed white noise, with $\delta=0.025$ 
(see Assumption~\ref{ass:noise}).

    \subsection{Policy nudged short-term \emph{vs} long-term shifts}
\begin{table*}[!tb]
    \caption{Social Benefit $\Gamma_{\mathrm{sim}}$ and cumulative cost $u_{\mathrm{sim}}^{\Sigma}$: \cite{c26} \emph{vs} proposed model and MPC strategy for different population bias and $L$.}
    \label{tab:comparison_c14}
\centering
    \begin{tabular}{|c|c|c|c|c|c|c|c|c|c|c|c|c|c|c|c|c|}
        \cline{2-17}
        \multicolumn{1}{c|}{} & \multicolumn{8}{c|}{Short-term shift model \cite{c26}} & \multicolumn{8}{c|}{Long-term shift model (ours)}\\
        \cline{2-17}
          \multicolumn{1}{c|}{} 
        & \multicolumn{4}{c|}{$\Gamma_{\mathrm{sim}}$ } 
        & \multicolumn{4}{c|}{$u_{\mathrm{sim}}^{\Sigma}$}
        & \multicolumn{4}{c|}{$\Gamma_{\mathrm{sim}}$  } 
        & \multicolumn{4}{c|}{$u_{\mathrm{sim}}^{\Sigma}$}
        \\

        \cline{2-17}
        \multicolumn{1}{c|}{} & \multicolumn{2}{c|}{\textbf{$L=5$}} & \multicolumn{2}{c|}{$L=20$} & \multicolumn{2}{c|}{$L=5$} & \multicolumn{2}{c|}{\textbf{$L=20$}} & \multicolumn{2}{c|}{$L=5$} & \multicolumn{2}{c|}{$L=20$} & \multicolumn{2}{c|}{\textbf{$L=5$}} & \multicolumn{2}{c|}{$L=20$}\\
        \hline 
        $\lambda$  & 0.25 & 0.75  & 0.25 & 0.75 & 0.25 & 0.75 & 0.25 & 0.75  & 0.25 & 0.75 & 0.25 & 0.75 & 0.25 & 0.75  & 0.25 & 0.75\\
        \hline 
          Mixed Bias & 6.32 & 5.50 & 6.17 & 5.34 & 9.58 & 10.00 & 7.49 & 8.25 & 0.01 & 0.01 & 0.01 & 0.01 & 9.50 & 9.50 & 9.50 & 9.50 \\
          \hline
         Negative Bias & 11.17 & 11.26 & 11.01 & 10.99 & 9.25 & 10.00 & 7.03 & 7.95 & 1.21 & 1.14 & 1.20 & 1.00 & 10.00 & 10.00 & 10.00  & 10.00 \\
        \hline
        Positive Bias & 1.91 & 1.85 & 1.83 & 1.78 & 9.73 & 10.00 & 7.74 & 8.67 & 0.01 & 0.01 & 0.01 & 0.01 & 5.50 & 5.50 & 5.50 & 5.50 \\
        \hline
    \end{tabular}
\end{table*}


Setting $\beta=10$ and imposing $R=10I$, we first compare the policies designed by solving \eqref{eq:MPC} in a receding horizon fashion for two different prediction horizons $L=5$ and $L=20$ and those obtained by tackling a similar problem, yet using the opinion dynamics model proposed in \cite{c26} for predictions and simulations. 

As clear from the values of the indicators reported in \tablename{~\ref{tab:comparison_c14}}, using the model proposed in \cite{c26} leads to more conservative policies when the considered horizon is longer. Indeed, for $L=20$, the MPC scheme relying on the model introduced in \cite{c26} is able to anticipate that the average inclination toward the technology/service will diminish when the available budget is exhausted, thus trying to slow the consumption of the budget. In turn, this results in a population that is, on average, less inclined to adopt the new technology/service compared to that achieved with the model and MPC strategy proposed in this work for the same $L$. Instead, when $L=5$ an MPC scheme relying on the model proposed in \cite{c26} suggests that the entire budget should be promptly used as the design strategy is \textit{short-sighted} with respect to the consequences of budget exhaustion (namely, that opinions will revert to their open-loop status). In contrast, the MPC strategy presented in Section~\ref{sec:MPC} maintains a performance that is consistent with that attained with a longer horizon, depleting the budget when needed depending on the initial agents' biases and ultimately achieving a more widespread positive inclination toward the technology/service of interest.

  
\subsection{Analyzing the impact of different budgets}
\begin{table*}[!tb]
    \caption{Social Benefit $\Gamma_{\mathrm{sim}}$ and budget consumption $B_{\%}$ for different budgets $\beta$.}
    \label{tab:budget_comparison}
    \centering
        \begin{tabular}{|c|c|c|c|c|c|c|c|c|c|c|c|c|}
        \cline{2-13}
        \multicolumn{1}{c|}{} & \multicolumn{4}{c|}{$\beta_\mathrm{high}$} & \multicolumn{4}{c|}{$\beta_\mathrm{mod}$}& \multicolumn{4}{c|}{$\beta_\mathrm{low}$}\\
        \cline{2-13}
        \multicolumn{1}{c|}{} & \multicolumn{2}{c|}{$\Gamma_{\mathrm{sim}}$} & \multicolumn{2}{c|}{$B_{\%}$} &\multicolumn{2}{c|}{$\Gamma_{\mathrm{sim}}$} & \multicolumn{2}{c|}{$B_{\%}$} & \multicolumn{2}{c|}{$\Gamma_{\mathrm{sim}}$} & \multicolumn{2}{c|}{$B_{\%}$} \\
        \hline 
        $\lambda$  & 0.25 & 0.75  & 0.25 & 0.75 & 0.25 & 0.75  & 0.25 & 0.75  & 0.25 & 0.75 & 0.25 & 0.75 \\
        \hline 
        Mixed bias  & 0.01 & 0.01 & 38.00 & 38.00 & 0.19 & 0.17 & 100.00 & 100.00 & 1.25 & 1.11 & 100.00 & 100.00 \\
        \hline
        Negative bias & 0.01 & 0.01 & 58.00 & 58.00 & 2.38 & 1.97 & 100.00 & 100.00 & 4.88 & 4.23 & 100.00 & 100.00 \\
        \hline
        Positive bias  & 0.01 & 0.01 & 22.00 & 22.00 & 0.01 & 0.01 & 68.75 & 68.75 & 0.05 & 0.04 & 100.00 & 100.00\\
        \hline
    \end{tabular}
\end{table*}
By considering the same penalty for the input effort introduced before, namely $R=10I$, we focus on the performance of the approach proposed in this paper for different budgets. In particular, we consider a scenario with a high budget $\beta_{\mathrm{high}}=25$, so that all available resources do not need to be depleted to achieve the widespread diffusion of the new technology/service, a setting with moderate budget $\beta_{\mathrm{mod}}=8$ (fully exhausted only for some of the scenarios we consider for the individual biases), and, lastly, a low budget $\beta_{\mathrm{low}}=5$ case, where all resources are depleted irrespective of individual biases.

\begin{figure}[!tb]
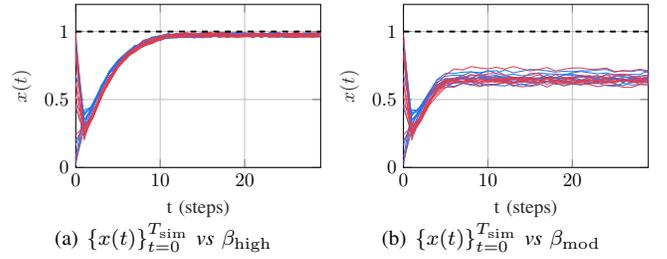

\centering
    \begin{tabular}{cc}
    \subfigure[$\{x(t)\}_{t=0}^{T_{\mathrm{sim}}}$ \emph{vs} $\beta_{\mathrm{high}}$]{\resizebox{0.485\linewidth}{!}{\input{Figures/fig3a}}} \hspace*{-.25cm}&\hspace*{-.25cm} \subfigure[$\{x(t)\}_{t=0}^{T_{\mathrm{sim}}}$ \emph{vs} $\beta_{\mathrm{mod}}$]{\resizebox{0.485\linewidth}{!}{\input{Figures/fig3b}}}  
    \end{tabular}\vspace{-.1cm}
    \caption{Negatively biased scenario: evolution of latent inclinations under different budget constraints.} \vspace{-.5cm}
    \label{fig:budgetcomp}
\end{figure}

As shown in \tablename{~\ref{tab:budget_comparison}}, the lower the available budget, the more individuals will be resistant to embrace the new technology on average. Moreover, the worst social benefit is achieved in the second scenario, as the population is negatively biased toward the technology/service. 
It can be observed that, when the budget is constrained, a higher social benefit is attained with $\lambda=0.75$. In this case, since the budget is limited, also social influence drives the promotion of the new technology/service, highlighting the importance of the interplay between social contagion and external policies in fostering innovation diffusion. Focusing on the negatively biased population with $\lambda=0.25$, \figurename{~\ref{fig:budgetcomp}} showcases the impact of budget restrictions on a realization of closed-loop inclinations. As expected, all individual inclinations approach $1$ when the budget is high (i.e., we have a near-universal acceptance of the technology), which is not the case when the available budget is moderate. Note that, in this scenario, the latent opinions at the end of the considered simulation horizon are nonetheless higher than the value achieved in open-loop (i.e., when no policy is enacted).

\subsection{Static \emph{vs} receding horizon policy}

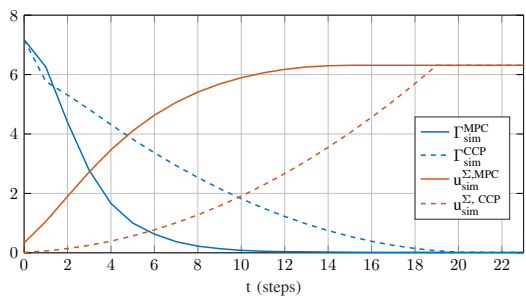
\begin{figure}[!tb]
\centering
    \resizebox{0.8\linewidth}{!}{
%
%
\definecolor{mycolor1}{rgb}{0.00000,0.44706,0.74118}%
\definecolor{mycolor2}{rgb}{0.85098,0.32549,0.09804}%
\begin{tikzpicture}

\begin{axis}[%
width=10.5cm,
height=5cm,
scale only axis,
xmin=0,
xmax=23,
xlabel style={font=\color{white!15!black}},
xlabel={t (steps)},
separate axis lines,
every outer y axis line/.append style={black},
every y tick label/.append style={font=\color{black}},
every y tick/.append style={black},
ymin=0,
ymax=8,
xmajorgrids,
ymajorgrids,
legend style={at={(0.782,0.14)}, anchor=south west, legend cell align=left, align=left, draw=white!15!black}
]
\addplot [color=mycolor1, line width=0.9pt]
  table[row sep=crcr]{%
0	7.17773798099997\\
1	6.25699993757846\\
2	4.39445858434321\\
3	2.77442860723662\\
4	1.66340143013384\\
5	1.00173482017827\\
6	0.629780687489213\\
7	0.375724521718046\\
8	0.223933622984593\\
9	0.14120667393092\\
10	0.0833914572026239\\
11	0.0532332077777432\\
12	0.0394827102760707\\
13	0.0349126285670455\\
14	0.0264126783337395\\
15	0.0198044341072709\\
16	0.0150925681994494\\
17	0.0136836082684125\\
18	0.016910500655988\\
19	0.0171948281815211\\
20	0.0124086814168254\\
21	0.0146988177160382\\
22	0.0119752389110469\\
23	0.0168209014442346\\
};
\addlegendentry{$\Gamma{}^{\text{MPC}}_{\text{sim}}$}

\addplot [color=mycolor1, line width=0.9pt, dashed]
  table[row sep=crcr]{%
0	7.17773798099997\\
1	5.76550146471698\\
2	5.31274622630861\\
3	4.82717019117416\\
4	4.31454911265365\\
5	3.82675322422694\\
6	3.36692943004544\\
7	2.93645699790426\\
8	2.53542228717828\\
9	2.16385909095184\\
10	1.82176863906168\\
11	1.50915155582238\\
12	1.22600778625225\\
13	0.97233732645334\\
14	0.748140169370151\\
15	0.553416313604253\\
16	0.388165758649158\\
17	0.252388504395144\\
18	0.146084550812213\\
19	0.0692538978938342\\
20	0.021896545638454\\
21	0.0145735245146091\\
22	0.0130017619762933\\
23	0.012624357101664\\
};
\addlegendentry{$\Gamma{}^{\text{CCP}}_{\text{sim}}$}

\addplot [color=mycolor2, line width=0.9pt]
  table[row sep=crcr]{%
0	0.333746448764137\\
1	1.05657338689741\\
2	1.89959940808566\\
3	2.72468409817535\\
4	3.46768601760126\\
5	4.1048356415281\\
6	4.63451194030828\\
7	5.06558265176844\\
8	5.41093607778502\\
9	5.6836418695552\\
10	5.89565174919931\\
11	6.05729852315241\\
12	6.17554510864001\\
13	6.25509150533166\\
14	6.29787348101698\\
15	6.31224701794804\\
16	6.31249830813538\\
17	6.3124986792247\\
18	6.31249966655014\\
19	6.31249989686498\\
20	6.3124998298263\\
21	6.31249998008036\\
22	6.31249987047073\\
23	6.31249996651732\\
};
\addlegendentry{$\text{u}^{\Sigma\text{,MPC}}_{\text{sim}}$}

\addplot [color=mycolor2, line width=0.9pt, dashed]
  table[row sep=crcr]{%
0	0.0157812504617348\\
1	0.0631250018469393\\
2	0.142031254155613\\
3	0.252500007387757\\
4	0.39453126154337\\
5	0.568125016622454\\
6	0.773281272625006\\
7	1.01000002955103\\
8	1.27828128740052\\
9	1.57812504617348\\
10	1.90953130586991\\
11	2.27250006648981\\
12	2.66703132803318\\
13	3.09312509050002\\
14	3.55078135389033\\
15	4.04000011820411\\
16	4.56078138344136\\
17	5.11312514960208\\
18	5.69703141668627\\
19	6.31250018469393\\
20	6.31250018469393\\
21	6.31250018469393\\
22	6.31250018469393\\
23	6.31250018469393\\
};
\addlegendentry{$\text{u}^{\Sigma\text{, CCP}}_{\text{sim}}$}

\end{axis}

\end{tikzpicture}%
}
    \caption{$\Gamma_{\mathrm{sim}}$ \emph{vs} $u_{\mathrm{sim}}^{\Sigma}$: constant \emph{vs} receding horizon policy.}\label{fig:comp_MPC_CP}
\end{figure}

Finally, we analyze the possible advantages of the receding horizon policy in \eqref{eq:MPC} compared to the constant policy \eqref{eq:CCP_opt} focusing on the mixed biases scenario, under a budget $\beta=10$ and imposing $R=15I$. Since the receding horizon policy exhausts the budget after $20$ time steps, we further set $T$ in \eqref{eq:CCP_opt} to $20$, while imposing $S=10$ to avoid waste of resources with the constant policy. 
As summarized in \figurename{~\ref{fig:comp_MPC_CP}}, within this setting, the MPC strategy uses more resources at the beginning than the CCP strategy, resulting in a greater social benefit in less time. The receding horizon strategy slightly outperforms the CCP one in terms of social benefit even at the end of the horizon, despite the cost of the two policies becoming aligned (and close to zero, as the whole budget is consumed).   

\section{Conclusions}\label{sec:conclusion}
In this work, we propose an opinion dynamics model that describes long-term shifts in opinion induced by external policies by introducing an artificial accumulation state, directly impacting
individual opinion dynamics. We then rely on the proposed model to introduce two strategies for policy design aimed at trading off a widespread adoption of a new (sustainable) technology/service and costs under budget constraints, whose impact on opinion dynamics is evaluated through numerical simulations. 

Future works will be devoted to blending models describing only short-term shifts in inclination with the proposed one, as well as analyzing the realism of such models and the validity of the proposed policy design strategies on real data.

\addtolength{\textheight}{-12cm}   


\bibliographystyle{IEEEtran}
\bibliography{Innovation_long_shifts.bib}

\end{document}